\documentclass[12pt]{amsart}

\usepackage{amsmath,amssymb,amsthm,a4,bm,epsfig,color,graphics}

\newtheorem{theorem}{Theorem}
\newtheorem{proposition}{Proposition}
\newtheorem{lemma}{Lemma}
\newtheorem{corollary}{Corollary}

\newtheorem*{remark*}{Remark}
\newtheorem*{remarks*}{Remarks}

\newcommand{\bbR}{{\mathbb R}}
\newcommand{\bbC}{{\mathbb C}}

\newcommand{\Hone}{H^{1}({\mathbb R}^N,{\mathbb C})}

\newcommand{\cL}{{\mathcal L}}
\newcommand{\cM}{{\mathcal M}}
\newcommand{\cN}{{\mathcal N}}
\newcommand{\cE}{{\mathcal E}}
\newcommand{\cT}{{\mathcal T}}

\newcommand{\cC}{{\mathcal C}}

\renewcommand{\Re}{{\mathrm{Re}}}
\renewcommand{\Im}{{\mathrm{Im}}}

\begin{document}

\title[]{Effective dynamics of solitons in the presence of rough nonlinear perturbations}

\author[]{W. K. Abou Salem}

\address{Department of Mathematics, University of Toronto, Toronto, Ontario, Canada M5S 2E4 \\ E-mail: walid@math.utoronto.ca} 
\maketitle

\begin{abstract}

The effective long-time dynamics of solitary wave solutions of the nonlinear Schr\"odinger equation in the presence of rough nonlinear perturbations is rigorously studied. It is shown that, if the initial state is close to a slowly travelling soliton of the unperturbed NLS equation (in $H^1$ norm), then, over a long time scale, the true solution of the initial value problem will be close to a soliton whose center of mass dynamics is approximately determined by an effective potential that corresponds to the restriction of the nonlinear perturbation to the soliton manifold. 

\end{abstract}


\section{Introduction}\label{sec:Introduction}

In this paper, we rigorously study the long time dynamics of solitary wave solutions of the nonlinear Schr\"odinger equation in the presence of {\it spatially rough} {\it nonlinear} {\it perturbations}. Physically, nonlinear perturbations of the NLS equation arise in inhomogeneous nonlinear optical media, and in Bose-Einstein experiments where the scattering length, which determines the nonlinear coefficient in the Gross-Pitaevskii equation, is modified using Feshbach resonances, such as in experiments involving atomic soliton lasers, see \cite{GA1} and references therein. 

Consider the initial value problem,
\begin{equation}
\label{eq:NLSE}
i\partial_t \psi = -\Delta \psi - |\psi|^{2s}\psi + \epsilon f(x,\psi),
\end{equation}
\begin{equation*}
\psi(t=0) =\varphi \in H^1(\bbR^N, \bbC), 
\end{equation*}
where $t\in \bbR$ denotes time, $x\in \bbR^N$ denotes a point in configuration space, with $N\ge 1,$ $\partial_t = \frac{\partial}{\partial t}$ the time derivative, $\Delta= \sum_{i=1}^N \frac{\partial^2}{\partial x_i^2}$ the $N$-dimensional Laplacian, $s\in (0,\frac{2}{N}),$ $\epsilon\in [0,1)$ and $f$ corresponds to a local nonlinear perturbation, which is a mapping on (complex) Sobolev spaces
\begin{equation*}
f: H^1(\bbR^N,\bbC)\rightarrow H^{-1}(\bbR^N,\bbC),
\end{equation*}
that is given by
\begin{equation}
\label{eq:Perturbation1}
f(x,\psi) = \lambda (x) |\psi|^{2\tilde{s}}\psi,
\end{equation}
where $\tilde{s}\in [0,\frac{2}{N-2})$ ($\tilde{s}\in [0,\infty)$ if $N=1,2,$), and $\lambda \in L^{\infty} (\bbR^N).$ \footnote{ Note that for a perturbation given by (\ref{eq:Perturbation1}), (\ref{eq:NLSE}) is globally well-posed in $H^1,$ see for example Theorem 6.1.1 in \cite{Ca1}; and also \cite{Su1}.} 

When $\epsilon=0,$ (\ref{eq:NLSE}) admits travelling soliton solutions, which are solutions of the form
\begin{equation*}
u(x,t) = e^{i\gamma + i \frac{1}{2} v\cdot (x-a)} \eta_\mu (x-a),
\end{equation*}
where $v\in \bbR^N$ is the velocity of the center of the soliton, $a=a_0 + v t \in \bbR^N$ denotes the center's position, $\gamma = \gamma_0 + \frac{v^2}{4}t \in \bbR$ is a phase, $\mu\in \bbR^+ ,$ and $\eta_\mu$ satisfies the nonlinear eigenvalue problem \footnote{See for example \cite{Ca1}, Chapter 8, and references therein, for an overview.}
\begin{equation*}
(-\Delta + \mu)\eta_\mu -\eta_\mu^3 =0.
\end{equation*}
 Starting with an initial condition which is close to a slowly travelling soliton, we want to study the effective dynamics of the true solution of the initial value problem if a rough nonlinear perturbation is switched on. We assume that $\exists \sigma_0 = (a_0,v_0,\gamma_0,\mu_0) \in \bbR^N\times \bbR^N \times [0,2\pi) \times \bbR^+$ and {\it fixed} $\alpha\in (0,1]$ such that
\begin{align}
\label{eq:InitialCondition1}
&\|\varphi - e^{i(\gamma_0 + \frac{1}{2}v_0(x-a_0))} \eta_{\mu_0} (x-a_0)\|_{H^1} < c  \ \  \epsilon^{\frac{1}{2}(1+\alpha)},\\
&\|v_0\| := (\sum_{i=1}^N v^2_{0,i})^{\frac{1}{2}} \le c \ \ \epsilon^\alpha, \nonumber
\end{align}
for some finite constant $c$ independent of $\epsilon$ and $\alpha.$ We have the following theorem, which is the main result of this paper.

\begin{theorem} 
\label{th:MainResult}
Consider the initial value problem (\ref{eq:NLSE}), and suppose that the nonlinear perturbation $f$ is given by (\ref{eq:Perturbation1}), 
and the initial condition $\varphi$ satisfies (\ref{eq:InitialCondition1}). 
Then there exists $\epsilon_0 >0,$ such that for all 
\begin{equation*}
\epsilon\in [0,\epsilon_0) , \ \ \beta\in (0,\alpha), \ \ \nu \in (0, \min(\beta,\alpha-\beta)),
\end{equation*} 
there exists an absolute positive constant $C$, independent of $\epsilon,\beta,$ and $\nu,$ yet dependent on the initial condition, such that,  for times $0\le t\le  C \nu |log \epsilon |/\epsilon^{\min (\beta-\nu, 1-\alpha)},$ the 
solution of the initial value problem (\ref{eq:NLSE}) is of the form 
\begin{equation*}
\psi(x,t) =e^{i(\frac{1}{2}v\cdot (x-a) + \gamma)}(\eta_\mu(x-a) + w(x-a,t)),
\end{equation*}
with
\begin{equation*}
\|w\|_{H^1}=O(\epsilon^{\frac{1}{2}(1+\tilde{\alpha})}),
\end{equation*}
where $\tilde{\alpha} = \alpha-\beta-\nu >0,$
and the parameters $a, v, \gamma$ and $\mu$ satisfy the 
differential equations
\begin{align*}
\partial_t a & = v + O(\epsilon^{1+\tilde{\alpha}}), \\
\partial_t v & = - 2 \sqrt{\mu} \ \ \nabla_a V_{eff,\mu}(a) + O(\epsilon^{1+\tilde{\alpha}}), \\ 
\partial_t \gamma & = \mu  + \frac{1}{4}v^2  -  \frac{2-Ns+2\tilde{s}}{2-Ns} V_{eff,\mu}(a) \\ & + \frac{s}{2-Ns} B_{eff,\mu}(a) + O(\epsilon^{1+\tilde{\alpha}}), \\ 
\partial_t{\mu} &= O(\epsilon^{1+\tilde{\alpha}}), 
\end{align*}
where 
\begin{equation*}
V_{eff, \mu}(a) = \frac{\epsilon\mu^{\frac{\tilde{s}}{s}}}{2+2\tilde{s}} \int ~ dx \lambda(\frac{x+a}{\sqrt{\mu}})\eta_1^{2\tilde{s}+2}(x),
\end{equation*}
and
\begin{equation*}
B_{eff,\mu}(a)= \frac{\epsilon\mu^{\frac{\tilde{s}}{s}}}{2+2\tilde{s}} \nabla_a \cdot \int ~ dx \lambda(\frac{x+a}{\sqrt{\mu}})\eta_1^{2\tilde{s}+2}(x) x .
\end{equation*}
\end{theorem}

We note that $V_{eff,\mu}$ corresponds to the restriction of the nonlinear perturbation to the soliton manifold, see Sect. \ref{sec:ProofMain}. In other words, for initial conditions close enough to a slowly travelling 
soliton, and for small enough nonlinear perturbations, the center of mass motion of the soliton is determined by Hamilton's (or Newton's) equations
of motion for a point particle in an external potential corresponding to the restriction of the perturbation to the soliton manifold, up to small errors due to radiation damping and the extended nature of the soliton. The proof of this result relies on three basic ingredients: First, using a skew-orthogonal decomposition, Subsect. \ref{sec:SkewOrtho}, the true solution of the NLS equation is decomposed into a part that belongs to the soliton manifold plus a fluctuation. By exploiting the group strucuture of the soliton manifold, Subsect. \ref{sec:GroupStructure}, and using a Lyapunov-Schmidt mapping onto the tangent space of the soliton manifold, (\ref{eq:SOProj}) in Subsect. \ref{sec:RepEqMotion}, it is shown that the dynamics of the part belonging to the soliton manifold is approximately determined by the restriction of the NLS equation to the soliton manifold, Corollary \ref{cr:ReparametrizedEqMotion} in Subsect. \ref{sec:RepEqMotion}. As for the fluctuation, its $H^1$ norm is controlled over a long time scale using an approximate Lyapunov exponent, Proposition \ref{pr:FluctuationBd} in Subsect. \ref{sec:ControlFluctuation}.

We now mention earlier results on solitary wave dynamics that are relevant to our analysis. In the last few years, there has been substantial progress in understanding solitary wave dynamics of the NLS equation with time-independent potentials in the semi-classical limit, \cite{BJ1,FTY1,FTY2,Ke1,FJGS1,FJGS2}. The latter analysis has been extended in \cite{HZ1} to studying the interaction of a slow soliton of the cubic nonlinear Schr\"odinger equation with a delta impurity in one dimension, and in \cite{A-S1} to studying the effective dynamics of  solitons in time-dependent potentials in the space-adiabatic limit. The basic picture is that if the external potential changes spatially slowly compared to the size of the soliton (or when it is small in the case of a delta impurity), and if the initial condition is close to a solitary wave solution, then, over a long time scale (related to the external potential and the initial condition), the true solution of the initial value problem is close to a soliton whose center of mass dynamics is approximately that of a point particle moving in an effective external potential. The novelty in this paper is studying the effective dynamics of solitons in the presence of  spatially rough and nonlinear perturbations, and also introducing in the analysis of soliton dynamics a useful Lyapunov-Schmidt mapping that is based on  the group structure of the soliton manifold. We rely on extending different parts of the analysis in \cite{FJGS1}-\cite{A-S1} to the case at hand, and also on other important developments in the theory of nonlinear Schr\"odinger equations, \cite{We1,We2,GSS1,GSS2}; see \cite{Ca1,Su1} for an overview. \footnote{We note that the problem of soliton dynamics in a slowly varying time-independent potential that was studied in \cite{FJGS1} is revisited in \cite{HZ2} for the special case of a cubic NLS equation in one dimension.} We hope that along the way, we clarify basic general concepts that are helpful in the study of long-time dynamics of solitons. \footnote{We note that the analysis presented below can be directly extended to the case of more general nonlinearities, such a sum of local and Hartree nonlinearities, and to  more general nonlinear perturbations, such as those with $\lambda$ slowly varying with respect to the size of the soliton, \cite{A-S2}.}

The organization of this paper is as follows. We first recall useful general properties of the nonlinear Schr\"odinger equation and the soliton manifold in Sect. \ref{sec:Revision}. In Sect. \ref{sec:ProofMain}, we prove the main result, Theorem \ref{th:MainResult}. 

\vspace{0.5cm}
\noindent {\bf Acknowledgements}
\vspace{0.5cm}

\noindent I thank J. Fr\"ohlich, G.-M. Graf, B.L.G. Jonsson, E. Lenzmann and I.M. Sigal for stimulating discussions. The partial financial support of NSERC grant NA 7901 is gratefully acknowledged. 


\section{Some properties of the NLS equation and the soliton manifold}\label{sec:Revision}

In this section, we recall some properties of the nonlinear Schr\"odinger equation and the soliton manifold. We will use these properties in the following sections. In what follows, we denote by 
\begin{equation}
\label{eq:LocalNL}
g(u):= |u|^{2s}u, \ \ u\in H^1(\bbR^N,\bbC),
\end{equation}
where $s\in (0,\frac{2}{N})$ appears in (\ref{eq:NLSE}), and we let the functionals $G,F:H^1(\bbR^N,\bbC)\rightarrow \bbR$ be such that $G'=g$ and $F'=f,$ where the prime corresponds to the Fr\'echet derivative. The nonlinear Schr\"odinger equation is 
\begin{equation}
\label{eq:NLSE2}
i\partial_t \psi = -\Delta \psi - g(\psi) + \epsilon f(x,\psi).
\end{equation}

\subsection{Symplectic, Hamiltonian and Variational structure}

The space $\Hone = H^1(\bbR^N,\bbR^2)$ as a real space, and it has a real inner product (Riemannian metric) 
\begin{equation}
\langle u,v \rangle :=\Re \int dx ~ u\overline{v} , 
\end{equation}
for $u,v\in \Hone,$ where $\overline{\cdot}$ stands for complex conjugation. \footnote{The tangent space ${\mathcal T}H^1=H^1.$} It is equipped with a symplectic form
\begin{equation}
\omega(u,v):= \Im \int  dx ~ u\overline{v} = \langle u,iv\rangle.
\end{equation}
The Hamiltonian functional corresponding to the nonlinear Schr\"odinger equation (\ref{eq:NLSE2}) is
\begin{equation}
\label{eq:NLHamiltonian}
H_\epsilon(\psi):=\frac{1}{2}\int |\nabla\psi|^2dx -G(\psi) +\epsilon F(\psi).
\end{equation}
Using the correspondence 
\begin{align*}
\Hone &\longleftrightarrow H^1(\bbR^N,\bbR) \oplus H^1(\bbR^N,\bbR)\\
\psi &\longleftrightarrow (\Re\psi,\Im\psi)\\
i^{-1}&\longleftrightarrow J,
\end{align*}
where $J:=\begin{pmatrix} 0 & 1 \\ -1 & 0 \end{pmatrix}$ is the complex structure on $H^1(\bbR^N,\bbR^{2}),$ the nonlinear Schr\"odinger equation (\ref{eq:NLSE2}) can be written as 
\begin{equation*}
\partial_t \psi = J H_\epsilon' (\psi).
\end{equation*}
In the following, we denote $H^1(\bbR^N)$ as either $H^1(\bbR^N,\bbC)$ or $H^1(\bbR^N,\bbR^2).$
We note that since the perturbation is time independent, the Hamiltonian functional $H_\epsilon$ defined in (\ref{eq:NLHamiltonian}) is autonomous, and energy is conserved, i.e., 
\begin{equation}
\label{eq:EnergyConservation}
H_\epsilon(\psi(t))=H_\epsilon(\psi(t=0)=\varphi), \ \ t\in\bbR.
\end{equation} 
Moreover, $H_\epsilon$ is invariant under global gauge transformations, 
\begin{equation*}
H_\epsilon(e^{i\gamma}\psi)=H_\epsilon(\psi),
\end{equation*}
and the associated conserved Noether charge is the ``mass'' \footnote{See for example \cite{Ca1}, Chapter 6, for a proof of (\ref{eq:EnergyConservation}) and (\ref{eq:ChargeConservation}).}
\begin{equation}
\label{eq:ChargeConservation}
N(\psi(t)):= \frac{1}{2}\int dx~ |\psi(t)|^2 = \frac{1}{2}\|\varphi\|_{L^2}^2 , \ \ t\in\bbR.
\end{equation}

As mentioned in Sect. \ref{sec:Introduction}, when $\epsilon=0,$ the nonlinear Schr\"odinger equation (\ref{eq:NLSE2}) admits travelling soliton solutions  that are of the form
\begin{equation*}
u(x,t) = e^{i\gamma + i \frac{1}{2} v\cdot (x-a)} \eta_\mu (x-a),
\end{equation*}
where $v\in \bbR^N$ is the velocity of the center of the soliton, $a=a_0 + v t \in \bbR^N$ denotes the center's position, $\gamma = \gamma_0 + \frac{v^2}{4}t \in \bbR$ is a phase, $\mu\in \bbR^+ ,$ and $\eta_\mu\in L^2(\bbR^N)\cap C^2(\bbR^N)$ is  positive, spherically symmetric and satisfies the nonlinear eigenvalue problem
\begin{equation}
\label{eq:NLEigenValue}
(-\Delta + \mu)\eta_\mu -\eta_\mu^3 =0,
\end{equation}
such that 
\begin{equation}
\label{eq:SolitonDecay}
\||x|^3\eta_\mu\|_{L^2}+\||x|^2|\nabla \eta_\mu|\|_{L^2} + \||x|^2\partial_\mu \eta_\mu\|_{L^2}<\infty, 
\end{equation}
$\forall \mu \in \bbR^+,$ and $\eta_\mu$ decays at $\infty$ as $e^{-\sqrt{\mu} \|x\|}.$ \footnote{See for example \cite{Ca1}, Chapter 8, and references therein.} Furthermore, the ``mass'' of the soliton is 
\begin{equation}
\label{eq:Mass}
m(\mu)=\frac{1}{2}\int dx ~\eta_\mu^2 = \mu^{\frac{1}{s}-\frac{N}{2}}.
\end{equation}
Note that
\begin{equation}
\label{eq:OrbitalStability}
\partial_\mu m(\mu) = (\frac{1}{s}-\frac{N}{2}) \mu^{\frac{1}{s}-\frac{N}{2}-1} >0
\end{equation} 
since $s\in (0,\frac{2}{N}),$ which implies orbital stability, \cite{GSS1,GSS2}.  
We define
\begin{equation}
\label{eq:Hessian}
\cL_\mu := -\Delta +\mu -g'(\eta_\mu),
\end{equation}
which is the Fr\'{e}chet derivative of the map $\psi\rightarrow (-\Delta+\mu)\psi -g(\psi)$ evaluated at $\eta_\mu.$ For all $\mu\in \bbR^+,$ the null space
\begin{equation*}
\cN(\cL_\mu) = span\{i\eta_\mu, \partial_{x_j}\eta_\mu, j=1,\cdots, N\},
\end{equation*}
see for example the Appendix in \cite{FJGS1}.

Orbital stability (\ref{eq:OrbitalStability}) implies that $\eta_\mu$ is a local minimizer of $H_{\epsilon=0}(\psi)$ restricted to the balls ${\mathcal B}_{m(\mu)}:= \{\psi\in H^1 : N(\psi)=m(\mu)\},$ where $m(\mu)$ is given in (\ref{eq:Mass}), see \cite{GSS1,GSS2}. They are critical points of the functional 
\begin{equation}
\label{eq:EnergyFunctional}
\cE_\mu (\psi) := \frac{1}{2} \int dx~ (|\nabla\psi|^2 +\mu|\psi|^2)-G(\psi).
\end{equation}

\subsection{Symmetries}

When $\epsilon=0,$ the nonlinear Schr\"odinger equation (\ref{eq:NLSE2}) is invariant under spatial translations
\begin{equation*}
T_a^{tr}: \psi(x,t)\rightarrow \psi(x-a,t), \ \ a\in \bbR^N,
\end{equation*}
spatial rotations 
\begin{equation*}
T_R^r: \psi(x,t)\rightarrow \psi(R^{-1}x,t), \ \ R\in SO(N),
\end{equation*}
time translations
\begin{equation*}
T_r^{t}: \psi(x,t)\rightarrow \psi(x,t-r),
\end{equation*}
gauge transformations
\begin{equation*}
T_\gamma^g : \psi(x,t)\rightarrow e^{i\gamma}\psi(x,t), \ \ \gamma \in [0,2\pi),
\end{equation*}
and Galilean transformations (boosts) 
\begin{equation*}
T_v^b: \psi(x,t)\rightarrow e^{\frac{i}{2}v\cdot x} \psi(x,t), \ \ v\in \bbR^N.
\end{equation*}
The conserved Noether quantities are the field momentum, angular momentum, energy, mass, and {\it center of mass motion},
\begin{equation*}
\int \bar{\psi}(-i\nabla)\psi , \ \ \int \overline{\psi} (x\wedge-i\nabla)\psi ,   \ \ \frac{1}{2}\int |\nabla\psi|^2 -G(\psi),  \ \ \frac{1}{2} \int |\psi|^2 , \ \ 
\int \bar{\psi} (x+2 i t \nabla)\psi .
\end{equation*}
Furthermore, when $\epsilon=0,$ (\ref{eq:NLSE2}) is invariant under complex conjugation
\begin{equation*}
T^c: \psi(x,t)\rightarrow \overline{\psi}(x,t)
\end{equation*}
and rescaling  
\begin{equation*}
T_\mu^s : \psi(x,t)\rightarrow \mu^{\frac{1}{2s}}\psi(\sqrt{\mu}x,\mu t).
\end{equation*}

\subsection{Soliton Manifold}\label{sec:SolitonManifold}

We introduce the {\it combined} transformation 
$$T_\sigma, \ \ \sigma =(a,v,\gamma,\mu)\in\bbR^N\times\bbR^N\times [0,2\pi)\times \bbR^+,$$
which is given by
\begin{equation}
\label{eq:CombinedTransformation}
T_\sigma \psi:= T_a^{tr} T_v^{b} T^g_\gamma T^s_\mu \psi = e^{i(\frac{1}{2}v\cdot (x-a)+\gamma)}\mu^{\frac{1}{s}}\psi(\sqrt{\mu}(x-a),\mu t),
\end{equation}
where $a,v \in\bbR^N,$ $\gamma\in [0,2\pi)$ and $\mu\in\bbR^+.$ We define the soliton manifold as 
\begin{equation}
\label{eq:SolitonManifold}
\cM_s:= \{\eta_\sigma = T_\sigma \eta_1 : \sigma = (a,v,\gamma,\mu) \in \bbR^N\times\bbR^N\times [0,2\pi)\times \bbR^+ \},
\end{equation}
where $\eta_1$ satisfies (\ref{eq:NLEigenValue}) with $\mu=1.$
The tangent space to the soliton manifold $\cM_s$ at $\eta_{\sigma_0}\in \cM_s$ is given by 
\begin{equation*}
\cT_{\eta_{\sigma_0}}\cM_s = span\{ z_t,z_g,z_b,z_s\},
\end{equation*}
where
\begin{align*}
z_t &:= \nabla_a T_a^{tr}\eta_{\sigma_0}|_{a=0}=-\nabla\eta_{\sigma_0}\\
z_g &:= \partial_\gamma T_\gamma^g \eta_{\sigma_0}|_{\gamma=0} = i\eta_{\sigma_0}\\
z_b &:= 2\nabla_v T_v^{b}\eta_{\sigma_0}|_{v=0}=ix\eta_{\sigma_0}\\
z_s &:= \partial_\mu T_\mu^s \eta_{\sigma_0}|_{\mu=1} = (\frac{1}{2s}+\frac{1}{2}x\cdot \partial_x)\eta_{\sigma_0}.
\end{align*}
In what follows, we denote 
\begin{align}
&e_i = -\partial_{x_i} , \ \ e_{i+N} = i x_i, \ \ i=1,\cdots ,N, \label{eq:Generators}\\
&e_{2N+1} = i ,\ \ e_{2N+2} = \frac{1}{2s}+\frac{1}{2}x\cdot \partial_x .\nonumber
\end{align}

The soliton manifold $\cM_s$ inherits a symplectic structure from $(H^1,\omega).$ For $\sigma = (a,v,\gamma,\mu) \in \bbR^N\times\bbR^N\times [0,2\pi)\times \bbR^+,$
\begin{equation*}
\Omega_\sigma^{-1} := J^{-1}|_{\cT_{\eta_\sigma}} = P_\sigma J^{-1} P_\sigma ,
\end{equation*}
where $P_\sigma$ is the $L^2$ orthogonal projection onto $\cT_{\eta_\sigma}\cM_s.$ It can be shown that $\Omega_\sigma^{-1}$ is invertible since $\partial_\mu m(\mu)>0,$ where the mass $m(\mu)=\frac{1}{2}\int dx ~\eta_\mu^2 = \mu^{\frac{1}{s}-\frac{2}{N}},$ see \cite{FJGS1}. Explicitly,
\begin{equation}
\label{eq:Metric}
\Omega_{\mu}^{-1}|_{\{ e_k \eta_\mu \}} := (\langle e_j\eta_\mu , J^{-1} e_k\eta_\mu\rangle)_{1\le j,k\le 2N+2} = 
\left(
\begin{matrix}
0 & -m(\mu) {\mathbf 1} & 0 & 0 \\
m(\mu) {\mathbf 1} & 0 & 0 & 0 \\
0 & 0 & 0 & m'(\mu) \\
0 & 0 & -m'(\mu) & 0
\end{matrix}
\right),
\end{equation}
where $e_k\eta_\mu, k=1,\cdots , 2N+2,$ are basis vectors of $\cT_{\eta_\mu} \cM_s,$ ${\mathbf 1}$ is the $N\times N$ identity matrix, and $m(\mu)=\mu^{\frac{1}{s}-\frac{N}{2}} ,  \ \ m'(\mu) = \partial_\mu m(\mu).$ Note that $\Omega_\sigma^{-1}$ is related to $\Omega_\mu^{-1}$ by a similarity transformation.

We note that the soliton $\eta_\mu$ breaks the translation and gauge symmetries of the nonlinear Schr\"odinger equation, which leads to associated zero modes of the {\it Hessian} $\cL_\mu,$ which is defined in (\ref{eq:Hessian}). Differentiating ${\mathcal E}_\mu' (T_a^{tr} T_\gamma^g \eta_\mu)=0$ with respect to $a$ and $\gamma$ and setting the latter two to zero gives 
\begin{equation}
\label{eq:ZeroModes1}
\cL_\mu z_t = 0 \; , \ \ \ \  \cL_\mu z_g = 0 \; , 
\end{equation}
while a direct computation gives
\begin{equation}
\label{eq:ZeroModes2}
\cL_\mu z_b = 2i z_t \; , \;  \cL_\mu z_s = i z_g .
\end{equation}

\subsection{Skew-Orthogonal Decomposition}\label{sec:SkewOrtho}

Consider the manifold $\cM_s' = \{\eta_\sigma, \ \ \sigma\in \Sigma_0\}, \ \ \Sigma_0 = \bbR^N\times \bbR^N \times [0,2\pi) \times I_0,$ where $I_0\subset I\backslash \partial I$ and $I\subset\bbR^+$ is bounded. We define the $\delta$ neighbourhood of $\cM_s'$ in $H^1$ as 
\begin{equation*}
U_\delta := \{ \psi\in H^1 , \ \ \inf_{\sigma\in \Sigma_0} \|\psi - \eta_\sigma\|_{H^1}\le \delta\}.
\end{equation*}
Then, for $\delta$ small enough and for all $\psi\in U_\delta,$ there exists a unique $\sigma(\psi)\in C^1(U_\delta,\Sigma)$ such that 
\begin{equation*}
\omega(\psi-\eta_{\sigma(\psi)}, z)= \langle \psi - \eta_{\sigma(\psi)}, J^{-1}z\rangle =0,
\end{equation*}
for all $z\in\cT_{\eta_{\sigma(\psi)}}\cM_s.$ For a proof of this statement, we refer the reader to \cite{FJGS1}, see also \cite{Ar1}.

\subsection{Group Structure}\label{sec:GroupStructure}
The action of the combined transformation $T_{\sigma=(a,v,\gamma,\mu)}$ defined in (\ref{eq:CombinedTransformation}) on elements of the soliton manifold $\cM_s$ has a group structure ${\mathcal G}/{\mathbb Z},$ where ${\mathcal G}$ corresponds to the semidirect product ${\mathsf H}^{2N+1}\ltimes \bbR^+,$ and ${\mathsf H}^{2N+1}$ is the real Heisenberg group in $2N+1$ dimensions. The quotient with ${\mathbb Z}$ is taken since $\gamma$ is defined modulo $2\pi.$ This group structure has been first noted in \cite{HZ1} for the case $N=1,$ but its generalization to higher dimensions is straightforward. The action of ${\mathcal G}$ is given by 
\begin{equation*}
(a',v',\gamma',\mu')\cdot(a,v,\gamma,\mu)=(a'',v'',\gamma'',\mu''),
\end{equation*}
where 
\begin{equation*}
a''=\frac{a}{\sqrt{\mu'}}+a', \ \ v''=\sqrt{\mu'} v+v'
\end{equation*}
\begin{equation*} 
\gamma''=\gamma+\gamma'+\frac{1}{2\sqrt{\mu'}}v\cdot a' , \ \ and \ \  \mu''=\mu \mu'.
\end{equation*}
The anti-selfadjoint operators $\{e_j\}_{j=1,\cdots,2N+2}$ defined in (\ref{eq:Generators}) form the generators of the corresponding Lie algebra ${\mathsf g}.$ They satisfy the commutation relations
\begin{align}
&[e_i,e_{j+N}]=-e_{2N+1}\delta_{ij}, \ \ i,j= 1,\cdots,N , \label{eq:CommutationRelations}\\
&[e_i,e_{2N+2}] = \frac{1}{2}e_i , \ \ i=1,\cdots,N, \nonumber\\
&[e_{i+N},e_{2N+2}] = -\frac{1}{2} e_{i+N}, \ \ i=1,\cdots,N,\nonumber 
\end{align}
and the rest of the commutators are zero. Note that it follows from (\ref{eq:SolitonManifold}) that $\cM_s$ is the orbit of ${\mathcal G}/{\mathbb Z},$ and given $z\in\cT_{\eta_\sigma}\cM_s,$ $\exists ! \ \ Y_z\in {\mathsf g}$ such that $z=Y_z \eta_\sigma.$ 
By exponentiating, the group element 
\begin{equation}
\label{eq:GroupElement}
T_{\sigma=(a,v,\gamma,\mu)} = e^{-a\cdot \partial_x} e^{i\frac{v\cdot x}{2}} e^{i\gamma} e^{\log (\mu) (\frac{1}{2s} + \frac{1}{2} x\cdot \partial_x)}.
\end{equation}


\section{Proof of the main result}\label{sec:ProofMain}

The proof of the main result is based on three ingredients, as discussed in the introduction. 

\subsection{Reparametrized equations of motion}\label{sec:RepEqMotion}

In this subsection, we use skew-orthogonal projection to decompose the solution of (\ref{eq:NLSE}) into a component that belongs to the soliton manifold plus a fluctuation. We then use the group structure and a Lyapunov-Schmidt mapping to obtain effective equations for the parameters characterizing the component belonging to the soliton manifold. 

Given the initial value problem (\ref{eq:NLSE}), let $T$ be the maximum time such that the skew-orthogonal decomposition, defined in Subsect. \ref{sec:SkewOrtho}, holds for $\psi(t), 0\le t \le T,$ for some $\delta>0,$ i.e., there exists a unique $\sigma(t)$ such that
\begin{equation*}
\psi(t) = T_{\sigma(t)} \eta_1 + w'(t),
\end{equation*}
such that 
\begin{equation*}
\omega(w'(t), Y T_{\sigma(t)}\eta_1) = 0 ,
\end{equation*}
for all $Y\in {\mathsf g}.$ In what follows, we denote $\eta_1$ by $\eta.$ Let $w(t) = T_{\sigma(t)}^{-1} w'(t).$ Then 
\begin{equation}
\label{eq:SODecomp}
\psi(t) = T_{\sigma(t)}u(t), \ \ u(t)=\eta + w(t) .
\end{equation}
It follows from (\ref{eq:CommutationRelations}) and the Baker-Campbell-Hausdorff formula, 
\begin{equation}
\label{eq:BCH}
e^X Y e^{-X} = e^{ad_X} Y, \ \ ad_X = [X,\cdot], 
\end{equation} 
for $X,Y\in {\mathsf g},$ that $T_\sigma YT_\sigma^{-1} \in {\mathsf g}$ if $Y\in {\mathsf g}.$ Furthermore, it follows from translational invariance that $\omega(T_\sigma u, T_\sigma v)=\mu^{\frac{1}{s}-\frac{N}{2}} \omega(u,v),$ for $u,v\in L^2,$ and hence 
\begin{align}
\label{eq:SO}
\omega(w(t), Y \eta) &=  \omega (T^{-1}_{\sigma(t)}w'(t), T^{-1}_{\sigma(t)} (T_{\sigma(t)} Y T^{-1}_{\sigma(t)}) T_{\sigma(t)} \eta ) \\ &=\mu^{-(\frac{1}{s}-\frac{N}{2})} \omega ( w'(t), Y'T_{\sigma(t)} \eta ) = 0, \nonumber 
\end{align}
$\forall Y\in {\mathsf g},$ where $Y'=T_{\sigma(t)} Y T^{-1}_{\sigma(t)}\in {\mathsf g}.$

We rewrite the nonlinear Schr\"odinger equation (\ref{eq:NLSE2}) as 
\begin{equation}
\label{eq:NLSE3}
\partial_t (T_{\sigma(t)} (\eta+w(t))) = -i H_\epsilon'(T_{\sigma(t)} (\eta+w(t)),
\end{equation}
where $H_\epsilon$ is defined in (\ref{eq:NLHamiltonian}) and the prime stands for the Fr\'echet derivative.

Suppose that, for the interval $[0,T],$ $$\{\sigma(t)=(a(t),v(t),\gamma(t),\mu(t))\}_{t\in [0,T]}$$ is a differentiable path in $\bbR^N\times\bbR^N\times \bbR\times \bbR^+.$ Then, differentiating $T_{\sigma(t)}$ with respect to $t$ and using (\ref{eq:CommutationRelations}) and the Baker-Campbell-Hausdorff formula (\ref{eq:BCH}) gives
\begin{equation}
\label{eq:DiffGroupElement}
\partial_t T_{\sigma(t)} = T_{\sigma(t)} \ \  X(\sigma(t)) 
\end{equation}
where $X(\sigma(t))\in {\mathsf g}$ is given by
\begin{equation}
\label{eq:X}
X(\sigma(t)) = \sum_{i=1}^N \dot{a}_i \sqrt{\mu } e_i + \sum_{i=1}^N \frac{\dot{v}_i}{2\sqrt{\mu}} e_{i+N} + (\dot{\gamma} -\frac{\dot{a} \cdot v}{2})e_{2N+1} + \frac{\dot{\mu}}{\mu} e_{2N+2},
\end{equation}
and the dot stands for $\partial_t.$

Using (\ref{eq:DiffGroupElement}), it follows that 
\begin{equation*}
\partial_t (T_{\sigma(t)} (\eta+w(t)) = T_{\sigma(t)}[X(\sigma(t))\eta + X(\sigma(t))w(t) + \partial_t w(t)],
\end{equation*}
and, together with (\ref{eq:NLSE3}), this implies
\begin{equation}
\label{eq:RateFluctuation1}
\partial_t w(t) = -X(\sigma(t))\eta - X(\sigma(t))w(t) -i T^{-1}_{\sigma(t)} H_\epsilon'(T_{\sigma(t)} (\eta+w(t)).
\end{equation}
Using (\ref{eq:NLHamiltonian}) and (\ref{eq:CombinedTransformation}) in Sect. \ref{sec:Revision}, a direct computation gives
\begin{align*}
T^{-1}_{\sigma(t)} H_\epsilon'(T_{\sigma(t)} (\eta+w(t)) = & \frac{v^2}{4}(\eta+w(t))-\mu [\Delta (\eta+w(t)) + |\eta+w(t)|^{2s}(\eta+w(t))]  \\ &-i\sqrt{\mu}v\cdot\partial_x(\eta+w(t))  +\epsilon \mu^{\frac{\tilde{s}}{s}} \lambda(\frac{x+a}{\sqrt{\mu}}) |\eta+w(t)|^{2\tilde{s}} (\eta+w(t)).
\end{align*}
Substituting back in (\ref{eq:RateFluctuation1}) and using (\ref{eq:NLEigenValue}) gives
\begin{equation}
\label{eq:RateFluctuation2}
\partial_t w = \tilde{X}(t)(\eta+w(t)) -i\mu \cL w(t) -i\epsilon \mu^{\frac{\tilde{s}}{s}} \lambda(\frac{x+a}{\sqrt{\mu}}) \eta^{2\tilde{s}+1} + i\mu \tilde{g}(\eta,w) -i\epsilon \tilde{f}(\frac{x+a}{\sqrt{\mu}},\eta,w), 
\end{equation}
where 
\begin{align}
\tilde{X}(t)& := -X(\sigma(t)) + \sum_{i=1}^N v_i\sqrt{\mu} e_i + (\mu-\frac{v^2}{4}) e_{2N+
1}\label{eq:TildeX} \\
&= \sum_{i=1}^N (v_i-\dot{a}_i) \sqrt{\mu } e_i + \sum_{i=1}^N \frac{-\dot{v}_i}{2\sqrt{\mu}} e_{i+N} + (-\dot{\gamma} +\frac{\dot{a} \cdot v}{2}- \frac{v^2}{4})e_{2N+1} - \frac{\dot{\mu}}{\mu} e_{2N+2},\nonumber
\end{align}
and 
\begin{equation*}
\cL = -\Delta +1 - g'(\eta)
\end{equation*}
the Hessian defined in (\ref{eq:Hessian}), evaluated at $\eta,$
\begin{equation}
\label{eq:TildeG}
\tilde{g}(\eta,w)= g(\eta+w(t)) - g(\eta) - g'(\eta)w(t),
\end{equation}
\begin{equation}
\label{eq:TildeF}
\tilde{f}(x,\eta,w) = \mu^{\frac{\tilde{s}}{s}}\lambda(x)[|\eta + w|^{2\tilde{s}}(\eta +w)-\eta^{2\tilde{s}+1}].
\end{equation}
Note that it follows from (\ref{eq:Perturbation1}), (\ref{eq:LocalNL}), (\ref{eq:TildeG}) and (\ref{eq:TildeF}) that 
\begin{equation}
\label{eq:TildeGUpperBd}
\|\tilde{g}(\eta,w)\|_{H^{-1}} \le C \|w\|^2_{H^1}
\end{equation}
and
\begin{equation}
\label{eq:TildeFUpperBd}
\|\tilde{f}(x,\eta,w)\|_{H^{-1}} \le C \|w\|_{H^1},
\end{equation}
uniformly in $t\in \bbR,$ where $C$ is a finite positive constant independent of $\epsilon.$

We now introduce the {\it Lyapunov-Schmidt mapping} $P: L^2 \rightarrow {\mathsf g}$ given by
\begin{equation}
\label{eq:SOProj}
P (\cdot):= \sum_{\alpha=1}^{2N+2} P_\alpha (\cdot) e_\alpha,
\end{equation}
where $e_\alpha, \alpha=1,\cdots ,2N+2,$ are given in (\ref{eq:Generators}) in Sec. \ref{sec:Revision}, and 
\begin{align}
\label{eq:Coefficients}
&P_i (\cdot) := \frac{-1}{m(\mu=1)} \omega (\cdot, e_{i+N}\eta) = -\omega (\cdot, e_{i+N}\eta), \ \ i=1,\cdots,N, \\
&P_{i+N}(\cdot) := \frac{1}{m(\mu=1)} \omega (\cdot, e_{i}\eta) = \omega (\cdot, e_{i}\eta) , \ \ i=1,\cdots,N  \nonumber \\
&P_{2N+1}(\cdot) := \frac{1}{m'(\mu=1)}\omega (\cdot, e_{2N+2}\eta) = \frac{2s}{2-sN} \omega (\cdot, e_{2N+2}\eta), \nonumber \\
&P_{2N+2} (\cdot) := \frac{-1}{m'(\mu=1)}\omega (\cdot, e_{2N+1}\eta) = -\frac{2s}{2-sN} \omega (\cdot, e_{2N+1}\eta). \nonumber 
\end{align}
This mapping is well-defined since it follows from (\ref{eq:SolitonDecay}) in Sec. \ref{sec:Revision} that $X\eta\in L^2$ for all $X = \sum_{\alpha=1}^{2N+2} X_\alpha e_\alpha \in {\mathsf g}$ such that $\|X\|:=\sup_{\alpha\in \{1,\cdots ,2N+2\}}|X_\alpha|<\infty.$ Note that for $X\in {\mathsf g},$ 
\begin{equation}
\label{eq:ProjId}
P(X\eta)=\sum_{\alpha,\beta=1}^{2N+2}X_\alpha P_\beta(e_\alpha \eta ) e_\beta = \sum_{\alpha,\beta=1}^{2N+2} X_\alpha \delta_{\alpha\beta} e_\beta = X.
\end{equation}

For $t\in [0,T],$ we introduce $Y(t)\in {\mathsf g}$ that is given by 
\begin{equation}
\label{eq:Y}
Y(t) : = \tilde{X}(t) -\epsilon \mu^{\frac{\tilde{s}}{s}} P(i\lambda(\frac{x+a}{\sqrt{\mu}}) \eta^{2\tilde{s}+1}),
\end{equation}  
where $\tilde{X}$ is given in (\ref{eq:TildeX}). The motivation for this choice of $Y$ will become clear below. 

\begin{proposition}\label{pr:ReparametrizedEq}
Suppose $\psi$ satisfies (\ref{eq:NLSE2}) such that $\psi(t)\in U_\delta,$ defined in Subsect. \ref{sec:SkewOrtho}, for $ t\in [0,T]$ and some $\delta >0,$ and let $w,\eta$ and $Y$ be as defined above. Then, for $\|w(t)\|_{H^1}<c,$ $c\in (0,1)$ independent of $\epsilon,$ we have 
\begin{equation*}
\|Y(t)\| \le C(\epsilon \|w(t)\|_{H^1} + \|w(t)\|_{H^1}^2),
\end{equation*}
where $\|Y\| := \sup_{\alpha\in\{ 1,\cdots,2N+2\} }|Y_\alpha|$ and $C$ is a positive constant independent of $t\in [0,T]$ and $\epsilon.$
\end{proposition}

\begin{proof}

It follows from (\ref{eq:SO}) and (\ref{eq:Coefficients}) that 
\begin{equation}
P(w(t)) = 0,
\end{equation}
and, since $P$ is independent of $t,$ $P(\partial_t w(t))=\partial_t P(w(t))=0.$ Together with (\ref{eq:RateFluctuation2}) and (\ref{eq:ProjId}), this implies 
\begin{align}
\label{eq:ProjRateFluctuation}
0&=P(\partial_t w(t)) \\ &= \tilde{X}(t)+P(\tilde{X}(t)w(t)) - \mu P(i\cL w(t)) - \epsilon \mu^{\frac{\tilde{s}}{s}} P(i\lambda(\frac{x+a}{\sqrt{\mu}}) \eta^{2\tilde{s}+1}) + \mu P(i\tilde{g}(\eta,w))  \nonumber \\ &- \epsilon P(i\tilde{f}(\frac{x+a}{\sqrt{\mu}},\eta,w)).\nonumber
\end{align}
Since $\cL$ is selfadjoint, it follows from (\ref{eq:ZeroModes1}) and (\ref{eq:ZeroModes2}) in Sec. \ref{sec:Revision} and (\ref{eq:SO}) that 
\begin{equation*}
\omega (i\cL w(t), e_\alpha\eta) = -\omega(w(t), i\cL e_\alpha\eta) = 0 ,\ \ \alpha=1,\cdots, 2N+2,
\end{equation*}
and hence 
\begin{equation}
\label{eq:ProjLFluctuation}
P(i\cL w(t))=0.
\end{equation}
Now, (\ref{eq:Y})-(\ref{eq:ProjLFluctuation}) imply 
\begin{equation}
\label{eq:ProjRateFluctuation2}
Y(t) = -P(\tilde{X}(t)w(t)) - \mu P(i\tilde{g}(\eta,w)) + \epsilon P(i\tilde{f}(\frac{x+a}{\sqrt{\mu}},\eta,w)).
\end{equation}
We want to find an upper bound for the absolute value of the right hand side of (\ref{eq:ProjRateFluctuation2}). In what follows, we denote by $C$ a positive constant that is independent of $\epsilon,$ and this constant changes from one line to another.

It follows from H\"older's inequality and (\ref{eq:SolitonDecay}) in Sec. \ref{sec:Revision} that
\begin{equation*}
|\omega (X w, e_\alpha \eta)| \le C \|X\| \|w\|_{L^2}, \ \ \alpha=1,\cdots ,2N+2,
\end{equation*}
and hence 
\begin{equation}
\label{eq:ProjXFluctuation}
\|P(X(t)w(t))\| \le C\|X\| \|w(t)\|_{L^2}.
\end{equation}
It also follows from (\ref{eq:SolitonDecay}) and (\ref{eq:TildeG}) that
\begin{equation*}
|\omega(i\tilde{g}(\eta,w(t)), e_\alpha \eta)| \le C \|w(t)\|_{H^1}^2, \ \ \alpha=1,\cdots ,2N+2,
\end{equation*}
and from (\ref{eq:SolitonDecay}) and (\ref{eq:TildeFUpperBd}) that 
\begin{equation*}
|\omega(i\tilde{f}(\frac{x+a}{\sqrt{\mu}},\eta,w(t)), e_\alpha\eta)| \le C \|w(t)\|_{H^1}, \ \ \alpha=1,\cdots ,2N+2.
\end{equation*}
Therefore, we have 
\begin{equation}
\|P(i\tilde{g}(\eta,w(t))\| \le C \|w(t)\|_{H^1}^2, \ \ \|P(i\tilde{f}(\frac{x+a}{\sqrt{\mu}},\eta,w(t)))\| \le C \|w(t)\|_{H^1}.
\label{eq:ProjPot}
\end{equation}
Now, (\ref{eq:ProjRateFluctuation2})-(\ref{eq:ProjPot}) give 
\begin{equation}
\|Y(t)\| \le C(\|Y(t)\|\|w(t)\|_{L^2} + \epsilon \|w(t)\|_{H^1} + \|w(t)\|_{H^1}^2),
\end{equation}
and hence the claim of the proposition. \end{proof}

We have the following corollary.
\begin{corollary}\label{cr:ReparametrizedEqMotion}
Suppose that the conditions of Proposition \ref{pr:ReparametrizedEq} hold. Then, for $t\in [0,T],$ 
\begin{align*}
\partial_t a & = v + O(\epsilon \|w(t)\|_{H^1} + \|w(t)\|_{H^1}^2), \\
\partial_t v & = - 2\sqrt{\mu} \ \ \nabla_a V_{eff,\mu}(a) + O(\epsilon \|w(t)\|_{H^1} + \|w(t)\|_{H^1}^2), \\ 
\partial_t \gamma & = \mu  + \frac{\partial_t a \cdot v}{2} - \frac{1}{4}v^2    
 -  \frac{2-Ns+2\tilde{s}}{2-Ns} V_{eff,\mu}(a) + \frac{s}{2-Ns} B_{eff,\mu}(a)\\ &+ O(\epsilon \|w(t)\|_{H^1} + \|w(t)\|_{H^1}^2), \\ 
\partial_t{\mu} &= O(\epsilon \|w(t)\|_{H^1} + \|w(t)\|_{H^1}^2), 
\end{align*}
where 
\begin{equation*}
V_{eff,\mu}(a) := \frac{\epsilon\mu^{\frac{\tilde{s}}{s}}}{2+2\tilde{s}} \int ~ dx \lambda(\frac{a+x}{\sqrt{\mu}})\eta^{2\tilde{s}+2}(x),
\end{equation*}
and
\begin{equation*}
B_{eff,\mu}(a) := \frac{\epsilon\mu^{\frac{\tilde{s}}{s}}}{2+2\tilde{s}} \nabla_a \cdot \int ~ dx \lambda(\frac{a+x}{\sqrt{\mu}})\eta^{2\tilde{s}+2}(x) x .
\end{equation*}
\end{corollary}

\begin{proof}
It follows from (\ref{eq:Y}) and Proposition \ref{pr:ReparametrizedEq} that 
\begin{equation}
\label{eq:Components0}
\|\tilde{X}(t) -\epsilon \mu^{\frac{\tilde{s}}{s}} P(i\lambda(\frac{x+a}{\sqrt{\mu}}) \eta^{2\tilde{s}+1})\| \le C (\epsilon\|w(t)\|_{H^1}+ \|w(t)\|_{H^1}^2),
\end{equation}
for some positive constant $C$ independent of $\epsilon$ and $t.$

Using (\ref{eq:SolitonDecay}), it follows by integration by parts that
\begin{align}
\label{eq:Components1}
\omega(i\tilde{f}(\frac{x+a}{\sqrt{\mu}},\eta,w(t)), -\partial_x \eta) &= \frac{1}{2\tilde{s}+2} \mu^{\frac{\tilde{s}}{s}} \partial_a \int ~ dx \lambda(\frac{x+a}{\sqrt{\mu}})\eta^{2\tilde{s}+2}(x), \\
\omega(i\tilde{f}(\frac{x+a}{\sqrt{\mu}},\eta,w(t)), (\frac{1}{2s}+\frac{1}{2}x\cdot\partial_x) \eta) &= (\frac{1}{2s}-\frac{N}{2\tilde{s}+2})\mu^{\frac{\tilde{s}}{s}} \int ~ dx \lambda(\frac{x+a}{\sqrt{\mu}})\eta^{2\tilde{s}+2}(x)  \\ &- \frac{1}{2+2\tilde{s}} \mu^{\frac{\tilde{s}}{s}} \partial_a \int ~ dx \lambda(\frac{x+a}{\sqrt{\mu}})\eta^{2\tilde{s}+2}(x) x.
\end{align}
Furthermore, since $\lambda$ and $\eta$ are real,
\begin{align}
&\omega(i\lambda(\frac{x+a}{\sqrt{\mu}}) \eta^{2\tilde{s}+1}, ix \eta) =0\\
&\omega(i\lambda(\frac{x+a}{\sqrt{\mu}}) \eta^{2\tilde{s}+1}, i\eta)=0.
\label{eq:Components2}
\end{align}
The proof follows directly from (\ref{eq:TildeX}), (\ref{eq:Coefficients}) and (\ref{eq:Components0})-(\ref{eq:Components2}). 
\end{proof}


\subsection{Control of the fluctuation}\label{sec:ControlFluctuation}

In this subsection, we use an approximate Lyapunov functional and the coercivity property of the Hessian $\cL_\mu$ to obtain an explicit control on the $H^1$ norm of the fluctuation. Our approach is a slight modification of that used in \cite{We1},\cite{We2}, \cite{FJGS1}-\cite{A-S1}. Recall that it follows from the skew-orthogonal decomposition, Subsect. \ref{sec:SkewOrtho}, that the solution of (\ref{eq:NLSE2}) can be rewritten as 
\begin{equation*}
\psi(t) = T_{av\gamma} (\eta_\mu(x) + w'(t,x)) \equiv T_a^{tr}T_v^{b}T_\gamma^{g} (\eta_\mu(x)+ w'(x,t)),
\end{equation*}
such that 
\begin{equation*}
\omega (w',Y\eta_\mu) = 0, \ \ \forall Y\in {\mathsf g}.
\end{equation*}
Here $w'=T_\mu^s w,$ where $w$ appears in (\ref{eq:SODecomp}) in Subsect. \ref{sec:RepEqMotion}.
Let $u':= \eta_\mu + w'.$ We define the Lyapunov functional 
\begin{equation}
\label{eq:LyapunovFunctional}
{\mathcal C}(\psi):= \cE_\mu (u') + \epsilon F(\psi) - (\cE_\mu(\eta_\mu) + \epsilon F(\eta_\mu(x-a))),
\end{equation}
where $\cE_\mu$ is defined in (\ref{eq:EnergyFunctional}), Sect. \ref{sec:Revision}, and $F'=f,$ where $f$ is given in (\ref{eq:Perturbation1}), Sect. \ref{sec:Revision}. We proceed by estimating upper and lower bounds for $\cC(\psi).$ 

\subsubsection{An upper bound for the Lyapunov functional}\label{sec:UBdLyapunov}

We have the following lemma.

\begin{lemma}\label{lm:UpperBoundLF}
Suppose $\psi$ satisfies (\ref{eq:NLSE2}) such that $\psi(t)\in U_\delta,$ for $t\in [0,T]$ and some $\delta >0,$ and let $u',w',\eta_\mu$ as above. Then, for $\|w'\|_{H^1}<1,$ there exists a constant $C$ independent of $\epsilon$ and $t$ such that 
\begin{equation}
\label{eq:UB1}
|{\mathcal C}(\psi(t))| \le C (y_0^2 + \epsilon y_t + t (\epsilon z_t + \epsilon  y_t^2 + \epsilon^2 y_t + y_t^4)),
\end{equation}
where $y_t:= \sup_{s\in [0,t]} \|w'(s)\|_{H^1}$ and $z_t:= \sup_{s\in [0,t]}\|v(s)\|.$ 
\end{lemma}

\begin{proof}
We have 
\begin{align*}
\cE_\mu (u') &= \frac{1}{2} \int dx |\nabla u'|^2 + \mu |u'|^2 -G(u') \\
&= H_\epsilon(u') + \frac{1}{2}\mu \|u'\|_{L^2}^2 - \epsilon F( u') \\
&= H_\epsilon(T_{av\gamma}^{-1}\psi) + \frac{1}{2}\mu \|T_{av\gamma}^{-1}\psi\|_{L^2}^2 - \epsilon F(T^{-1}_{av\gamma}\psi).
\end{align*}
By translational symmetry, 
\begin{equation*}
\|T_{av\gamma}^{-1}\psi\|_{L^2}^2 = \|\psi\|_{L^2}^2 .
\end{equation*}
Furthermore, 
\begin{equation*}
H_\epsilon(T_{av\gamma}^{-1}\psi) = H_\epsilon(\psi)+ \frac{1}{8}v^2 \|\psi\|_{L^2}^2 - \frac{1}{2}v\cdot \langle i\psi, \nabla \psi \rangle + \epsilon F(T_{av\gamma}^{-1}\psi) - \epsilon F(\psi),
\end{equation*}
and hence 
\begin{equation}
\label{eq:CMEnergy}
\cE_\mu (u')= H_\epsilon(\psi) + \frac{1}{2}(\frac{1}{4}v^2 + \mu)\|\psi\|_{L^2}^2 - \frac{1}{2}v\cdot \langle i\psi ,\nabla \psi \rangle - \epsilon F(\psi).
\end{equation}
Recall that since the perturbation $f$ is time independent, energy is conserved,
\begin{equation}
\label{eq:RateEnergy1}
\partial_t H_\epsilon(\psi) = 0,
\end{equation}
while the rate of change of momentum is 
\begin{equation}
\label{eq:Ehrenfest1}
\partial_t \langle\psi,-i \nabla \psi\rangle =  \epsilon\langle f(x,\psi),\nabla \psi\rangle + \epsilon\langle \nabla\psi, f(x,\psi)\rangle.
\end{equation}
Formally, (\ref{eq:Ehrenfest1}), which is a statement of a generalized Ehrenfest Theorem, follows from (\ref{eq:NLSE2}). To prove it, we introduce the regularizing operator $I_h:= (1-h\Delta)^{-1}.$ For properties of $I_h$ as $h\rightarrow 0,$ we refer the reader to Proposition 2.4.2 in \cite{Ca1}. Using (\ref{eq:NLSE2}), we have 
\begin{align*}
\partial_t \langle i\psi, \nabla \psi\rangle & = \lim_{h\rightarrow 0} \{\langle I_h i\partial_t\psi, I_h\nabla \psi\rangle - \langle I_h\psi,  \nabla I_h i\partial_t\psi\rangle\} \\
&= \lim_{h\rightarrow 0} \{\langle I_h( -\Delta \psi - g(\psi) + \epsilon f(x,\psi)), I_h \nabla\psi\rangle \\ & - \langle I_h\psi,  \nabla I_h (-\Delta \psi - g(\psi) + \epsilon f(x,\psi))\rangle\} \\
&= \epsilon \{ \langle f(x,\psi), \nabla \psi \rangle + \langle \nabla \psi, f(\psi)\rangle\} .
\end{align*}

We also have 
\begin{align*}
&\partial_t \{\frac{1}{2}(\frac{v^2}{4}+\mu)\|\psi\|_{L^2}^2 - \frac{v}{2} \langle i\psi, \nabla \psi\rangle\} =  
- \frac{v}{2} \partial_t \langle i\psi, \nabla \psi\rangle -\frac{\partial_t v}{2} \langle i\psi, \nabla \psi\rangle +\frac{v\partial_t v + 2\partial_t \mu}{4}\|\psi\|_{L^2}^2 \\
&=  - \frac{v}{2} \partial_t \langle i\psi, \nabla \psi\rangle -\frac{\partial_t v}{2} \langle i w' , \nabla w' \rangle + \frac{\partial_t \mu}{2}(\|\eta_\mu\|_{L^2}^2 + \|w'\|_{L^2}^2),
\end{align*}
where we have used in the first line the fact that $\partial_t \|\psi\|_{L^2}=0$ (charge conservation), and in the last equality the facts that $\psi=e^{\frac{i}{2}v\cdot (x-a) + i\gamma}(\eta_\mu + w'),$ $\langle i \nabla \eta_\mu , w' \rangle = \langle i w' ,\nabla \eta_\mu\rangle =0,$ and $\langle i\eta_\mu, iw'\rangle = \langle iw',i\eta_\mu\rangle =0$ (which follow from skew-orthogonal decomposition). 
Furthermore, since $\eta_\mu$ is a minimizer of $\cE_\mu,$
\begin{equation*}
\partial_t \cE_\mu(\eta_\mu) = \frac{1}{2}\partial_t\mu \|\eta_\mu\|^2_{L^2}.
\end{equation*}

It follows that 
\begin{align*}
\partial_t {\mathcal C}(\psi) &= \partial_
t H_\epsilon (\psi) -\frac{1}{2}\epsilon v \cdot \{\langle f(\psi),\nabla \psi\rangle + \langle \nabla\psi, f(\psi)\rangle \} \\ &-\frac{1}{2}\partial_t v \cdot \langle i\psi,\nabla\psi\rangle +\frac{1}{2} (\frac{v\partial_t v}{2} + \partial_t\mu) \|\psi\|_{L^2}^2 + \frac{1}{2} (\frac{v^2}{4} +\mu)\partial_t \|\psi\|_{L^2}^2\\ & -\frac{1}{2} \partial_t \mu \|\eta_\mu\|_{L^2}^2 - \epsilon \partial_ta \cdot\nabla_a F(\eta_\mu(x-a)) - \epsilon \partial_t \mu \partial_\mu F(\eta_\mu (x-a)) \\
&= -\frac{1}{2}\epsilon v \cdot \{\langle f(\psi),\nabla \psi\rangle + \langle \nabla\psi, f(\psi)\rangle \} -\frac{\partial_t v}{2} \langle i w' , \nabla w' \rangle \\ & + \frac{1}{2}\partial_t\mu \|w'\|_{L^2}^2  - \epsilon \partial_ta \cdot\nabla_a F(\eta_\mu(x-a)) - \epsilon \partial_t \mu \partial_\mu F(\eta_\mu (x-a))
\end{align*}
Together with (\ref{eq:SolitonDecay}), Corollary \ref{cr:ReparametrizedEqMotion}, and the fact that $\|w\|_{H^1}= \|T^s_{\mu^{-1}}w'\|_{H^1}\le C\|w'\|_{H^1}$, it follows that 
\begin{equation}
\label{eq:RateUBd1}
|\partial_t {\mathcal C}(\psi)| \le C (\epsilon \|v\| +\epsilon \|w'\|_{H^1}^2 +\epsilon^2 \|w'\|_{H^1}+ \|w'\|_{H^1}^4),
\end{equation}
where $C$ is a positive constant independent of $\epsilon.$ 
Furthermore, by expanding $\cE_{\mu_0}(\eta_{\mu_0}+w_0')$ around the minimizer $\eta_{\mu_0},$ we have 
\begin{equation}
\label{eq:IntDiff1}
|\cE_{\mu_0}(\eta_{\mu_0}+w_0')-\cE_{\mu_0} (\eta_{\mu_0})| \le C \|w_0'\|^2_{H^1}.
\end{equation}
We also have 
\begin{align}
\label{eq:PertDiff1}
|F(\psi) - F(\eta_\mu(x-a))| &= |F(\eta_\mu(x-a) + w' (t,x-a)) - F(\eta_\mu(x-a))| \\ &\le C \|w'\|_{H^1}.\nonumber
\end{align}
Now, (\ref{eq:RateUBd1}) - (\ref{eq:PertDiff1}) together with the fundamental theorem of Calculus imply the claim of the lemma.
\end{proof}

\subsubsection{A lower bound for the Lyapunov functional}

In this subsection, we estimate a lower bound for $\cC (\psi).$ Let 
\begin{equation*}
X_\mu := \{w\in H^1(\bbR^N) : \ \ \|w\|_{H^1}=1, \ \ \langle w, J^{-1} z\rangle =0, \forall z\in \cT_{\eta_\mu}\cM_s \}.
\end{equation*}
It follows from the coercivity property of $\cL_\mu$ that there exists a positive constant 
\begin{equation}
\label{eq:Coercivity}
\rho:= \inf_{w\in X_\mu}\langle w,\cL_\mu w\rangle >0.
\end{equation}
We refer the reader to Appendix D in \cite{FJGS1} for a proof of this statement. 

\begin{lemma}\label{lm:LFLowerBound}
Suppose $\psi$ satisfies (\ref{eq:NLSE2}) such that $\psi(t)\in U_\delta, \ \ t\in [0,T],$ for some $\delta >0,$ and let $u',w',\eta_\mu$ be as defined as above. Then there exists a positive constant  $C$ independent of $\epsilon$ such that, for $\|w'\|_{H^1}< 1,$ 
\begin{equation}
\label{eq:LFLowerBound}
|\cC (\psi(t))| \ge \frac{\rho}{2} \|w'(t)\|_{H^1}^2 - C(\|w'(t)\|_{H^1}^3 +  \epsilon \|w'(t)\|_{H^1}),
\end{equation}
where $\rho$ appears in (\ref{eq:Coercivity}), uniformly in $t\in [0,T].$
\end{lemma}

\begin{proof}

Expanding $\cE_\mu(u')$ around the minimizer $\eta_\mu,$ we have
\begin{equation}
\cE_\mu (\eta_\mu + w') = \cE_\mu (\eta_\mu) + \frac{1}{2}\langle w', \cL_\mu w'\rangle + R_\mu^{(3)}(w'),
\end{equation}
where 
\begin{equation*}
R_\mu^{(3)}(w')= G(\eta_\mu +w') - G(\eta_\mu)-\langle G'(\eta_\mu),w'\rangle -\frac{1}{2} \langle G''(\eta_\mu)w',w'\rangle.
\end{equation*}
The claim of the lemma follows from  
\begin{equation*}
|R_\mu^{(3)}(w')| \le C \|w'\|_{H^1}^3,
\end{equation*}
\begin{equation*}
|F(\psi) - F(\eta_\mu(x-a))| \le C \|w'\|_{H^1}.
\end{equation*}
and the coercivity property (\ref{eq:Coercivity}).
\end{proof}


\subsubsection{Bound on the fluctuation}\label{sec:BdFluctuation}

In this subsection, we combine both the upper and lower bounds on the Lyapunov functional and use a bootstrap argument to obtain a bound on the $H^1$ norm of the fluctuation. 

\begin{proposition}\label{pr:FluctuationBd}
Consider the initial value problem (\ref{eq:NLSE}), and suppose that the nonlinear perturbation $f$ is given by (\ref{eq:Perturbation1}), 
and the initial condition $\varphi$ satisfies (\ref{eq:InitialCondition1}). 
Then, there exists $\epsilon_0 >0$ that depends on the initial condition, such that, for all $\epsilon\in [0,\epsilon_0) , $  $\beta\in (0,\alpha)$ and $\nu \in (0, \min(\beta,\alpha-\beta)),$ there exists absolute positive constants $C, C'$, independent of $\epsilon,\beta,$ and $\nu$ such that,  for times $$0\leq t\leq  C \nu \frac{|log \epsilon |}{\max (\epsilon^{\beta-\nu}, \epsilon^{1-\alpha})},$$ 
\begin{equation}
\|w'(t)\|_{H^1}^2 \le C' \epsilon^{1+\alpha-\beta-\nu}.
\end{equation}
\end{proposition}

\begin{proof}
Suppose  $\psi(t)\in U_\delta, \ \ t\in [0,T],$ for some $\delta >0,$ and let $u',w',\eta_\mu$ be as defined as above. It follows from Lemmata \ref{lm:UpperBoundLF} and \ref{lm:LFLowerBound} that, for $y_t<1,$ there exists a constant $C$ independent of $\epsilon$ such that
\begin{equation}
\label{eq:Bd2}
y_t^2 \le C (y_0^2 + \epsilon y_t + y_t^3 + t (\epsilon z_t + \epsilon y_t^2 + \epsilon^2 y_t + y_t^4)) 
\end{equation}
where $y_t= \sup_{s\in [0,t]} \|w'(s)\|_{H^1}$ and $z_t= \sup_{s\in [0,t]}\|v(s)\|.$
It also follows from Corollary \ref{cr:ReparametrizedEqMotion} that
\begin{equation}
\label{eq:vBd1}
z_t \le z_0 + Ct (\epsilon + y_t^2 + \epsilon y_t).
\end{equation}
We choose $\beta\in (0,\alpha),$ where $\alpha$ appears in (\ref{eq:InitialCondition1}), and we define
\begin{equation}
\label{eq:Tau}
\tau_t:= \min(\frac{\epsilon^\alpha}{\epsilon + y_t^2 + \epsilon y_t},\frac{\epsilon^{1+\alpha-\beta}}{\epsilon z_t + \epsilon y_t^2 + \epsilon^2 y_t + y_t^4}).
\end{equation}
For $t\le \tau_t,$ (\ref{eq:Bd2}) and (\ref{eq:vBd1}) imply that
\begin{align}
\label{eq:Bd3}
&y_t^2 \le C' (y_0^2 + \epsilon y_t + y_t^3 + \epsilon^{1+\alpha-\beta}) \\
&z_t \le C' (z_0+\epsilon^\alpha ),\nonumber
\end{align}
for some positive constant $C'$ that is independent of $\epsilon$ and $\beta.$
For $\epsilon$ small enough, (\ref{eq:InitialCondition1}) and the first inequality in (\ref{eq:Bd3}) imply that 
\begin{equation}
y_t^2 \le C (y_0^2 + \epsilon^{1+\alpha - \beta}),
\end{equation}
for some positive constant $C$ that is independent of $\epsilon$ and $\beta.$ It follows that
\begin{align}
\label{eq:Bd4}
&y_t^2 \le C (y_0^2 + \epsilon^{1+\alpha-\beta}) \\
&z_t \le C (z_0+\epsilon^\alpha ),\nonumber
\end{align}
for an absolute positive constant $C>1$ that is independent of $\epsilon$ and $\beta,$ uniformly in $t$ such that $0\le t \le\tau_t.$ 

We now reiterate the above analysis. Consider the interval 
\begin{equation*}
[0,T'] = [t_0,t_1]\cup [t_1,t_2] \cup \cdots \cup [t_{n-1},t_n] \subset [0,T],
\end{equation*}
such that 
$$0=t_0< t_1 < \cdots < t_n=T', \ \ (t_{i}-t_{i-1}) \le \tau_{t_{i}}, \ \ i=1,\cdots, n.$$ Let 
\begin{align*}
&y_i := \sup_{t\in [t_{i-1},t_{i}]} \|w'(t)\|_{H^1}, \ \ i=1,\cdots , n.\\
&z_i := \sup_{t\in [t_{i-1},t_{i}]} \|v(t)\|. 
\end{align*}
Note that from (\ref{eq:InitialCondition1}), $y_0 \le C'\epsilon^{\frac{1}{2}(1+\alpha)}$ and $z_0 \le C'\epsilon^\alpha,$ for some constant $C'$ independent of $\epsilon.$ Iterating (\ref{eq:Bd4}) $n$ times, we have 
\begin{align}
\label{eq:Bd5}
&y_n^2 \le (\sum_{j=1}^n C^j) C'\epsilon^{1+\alpha-\beta} \le C^{n+1} C' \epsilon^{1+\alpha-\beta}  \\
&z_n \le (\sum_{j=1}^n C^j)C' \epsilon^{\alpha} \le C^{n+1} C' \epsilon^{\alpha}.\nonumber
\end{align}
We choose 
\begin{equation}
\label{eq:Nu}
\nu\in (0, \min(\beta,\alpha-\beta)).
\end{equation}
Given $\nu$ and $C,$ we choose $n$ such that
\begin{equation}
\label{eq:Bd6}
C^{n+1}\le \epsilon^{-\nu}. 
\end{equation}
This implies 
\begin{equation*}
n \le -\nu \frac{\log \epsilon}{\log C}-1. 
\end{equation*}
It follows from (\ref{eq:Bd5}) and (\ref{eq:Bd6}) that 
\begin{align}
\label{eq:Bd7}
&y_n^2 \le C'\epsilon^{1+\alpha-\beta-\nu},\\
&z_n \le C'\epsilon^{\alpha-\nu}.
\end{align}
We define 
\begin{equation}
\label{eq:TauPrime}
\tau':= \frac{1}{4C'}\min (\frac{1}{\epsilon^{\beta-\nu}}, \frac{1}{\epsilon^{1-\alpha}}).
\end{equation}
One can directly verify using (\ref{eq:Tau}), (\ref{eq:Nu}), (\ref{eq:Bd7}) and (\ref{eq:TauPrime}) imply that that
\begin{equation*}
\tau'\le \tau_{t_i}, \ \ i=1,\cdots ,n,
\end{equation*}
for $\epsilon$ small enough. Therefore, $\exists\epsilon_0 >0$ such that if $\epsilon<\epsilon_0,$ there exists absolute constants $C$ and $C'$ independent of $\epsilon,\beta$ and $\nu$ such that 
\begin{align*}
&y_t^2 \le C' \epsilon^{1+\alpha-\beta-\nu}, \\
&z_t \le C' \epsilon^{\alpha-\nu},
\end{align*}
uniformly for $t\in [0,\nu C |\log\epsilon|/\epsilon^{\min (\beta-\nu , 1-\alpha)} ],$ and hence the claim of the proposition.
\end{proof}

\subsection{Proof of Theorem \ref{th:MainResult}.}

\begin{proof}
The proof follows directly from Corollary \ref{cr:ReparametrizedEqMotion} and Proposition \ref{pr:FluctuationBd}. 
\end{proof}



\begin{thebibliography}{9}


\bibitem{GA1}
J. Garnier and F.Kh. Abdullaev.
\newblock  Transmission of matter wave solitons through nonlinear traps and barriers.
\newblock {\em Phys. Rev. A}, 74: 013604, 2006.


\bibitem{Ca1}
T. Cazenave.
\newblock {\em An Introduction to Nonlinear Schr\"odinger Equations.}
\newblock Textos de M\'etodos Matem\'aticos 26. Instituto de Matem\'atica, Rio de Janeiro, 1996.

\bibitem{Su1} 
C. Sulem and P.-L. Sulem.
\newblock {\em The Nonlinear Schr\"odinger Equation.}
\newblock Number 130 in Applied Mathematical Sciences. Springer, New York, 1999.


\bibitem{BJ1}
J.~C. Bronski and R.~L. Jerrard.
\newblock Soliton dynamics in a potential.
\newblock {\em Math. Res. Lett.,} 7(2-3): 329-342, 2000.


\bibitem{FTY1}
J.~Fr\"ohlich, T.-P. Tsai, and H.-T. Yau.
\newblock On a classical limit of quantum theory and the non-linear Hartree
  equation.
\newblock {\em Geom. Funct. Anal.}, Special Volume: 57-78, 2000.

\bibitem{FTY2}
J.~Fr\"ohlich, T.-P. Tsai, and H.-T. Yau.
\newblock On the point-particle (Newtonian) limit of the non-linear Hartree equation.
\newblock {\em Commun. Math. Phys.}, 225(2): 223-274, 2002.


\bibitem{Ke1}
S. Keraani. 
\newblock Semiclassical limit of a class of Schr\"odinger equations with potential.
\newblock {\em Comm. Part. Diff. Eq.}, 27: 693--704, 2002.


\bibitem{FJGS1}
J.~Fr\"ohlich, S.~Gustafson, B.~L.~G. Jonsson, and I.~M. Sigal.
\newblock Solitary wave dynamics in an external potential.
\newblock {\em Commun. Math. Phys.}, 250(3): 613-642, 2004.

\bibitem{FJGS2}
J.~Fr\"ohlich, S.~Gustafson, B.~L.~G. Jonsson, and I.~M. Sigal.
\newblock Long time motion of NLS solitary waves in a confining potential. \newblock {\em Annals Henri Poincare,} 7: 621-660, 2006.



\bibitem{HZ1}
J. Holmer and M. Zworski.
\newblock Slow soliton interaction with delta impurities.
\newblock {\it Preprint} 2007.



\bibitem{A-S1}
W. K. Abou Salem.
\newblock Solitary wave dynamics in time-dependent potentials.
\newblock {\it Preprint} 2007.

\bibitem{We1}
M.~I. Weinstein.
\newblock Modulational stability of ground states of nonlinear Schr\"{o}dinger equations.
\newblock {\em SIAM J. Math. Anal.}, 16(3): 472-491, 1985.

\bibitem{We2} 
M.~I. Weinstein.
\newblock Lyapunov stability of ground states of the nonlinear dispersive evolution equations.
\newblock {\em Commun. Pure Appl. Math.,} XXXIX:51-68, 1986.


\bibitem{GSS1}
M.~Grillakis, J.~Shatah, and W.~Strauss.
\newblock Stability theory of solitary waves in the presence of symmetry. {I}.
\newblock {\em J. Funct. Anal.}, 74(1): 160-197, 1987.

\bibitem{GSS2}
M.~Grillakis, J.~Shatah, and W.~Strauss.
\newblock Stability theory of solitary waves in the presence of symmetry. {II}.
\newblock {\em J. Funct. Anal.}, 94(2): 308-348, 1990.


\bibitem{HZ2}
J. Holmer and M. Zworski.
\newblock Soliton interaction with slowly varying potentials.
\newblock {\it Preprint} 2007.


\bibitem{A-S2} In preparation.



\bibitem{Ar1}
V.~Arnold. 
\newblock {\em Mathematical Methods of Classical Mechanics}.
\newblock Springer-Verlag, Graduate texts in Mathematics, New York, 1984. 


\end{thebibliography}
\end{document}